\documentclass[preprint,12pt]{elsarticle}

\usepackage{amssymb}
\usepackage{amsmath}
\usepackage{amsthm}
\usepackage{algorithm}
\usepackage{xurl}
\usepackage{microtype}
\usepackage[hidelinks]{hyperref}
\ifdefined\pdfminorversion
\fi

\newtheorem{lemma}{Lemma}
\newtheorem{proposition}{Proposition}
\newtheorem{theorem}{Theorem}
\newcommand{\symdiff}{\mathbin{\triangle}}
\providecommand{\doi}[1]{\href{https://doi.org/#1}{\nolinkurl{doi:#1}}}

\journal{Discrete Applied Mathematics}

\begin{document}

\begin{frontmatter}

\title{A \texorpdfstring{$2/3$}{2/3}-Approximation Algorithm for the Residual Maximum Matching
Problem in General Graphs}

\author[inst1]{Quanlong Niu}

\author[inst2]{Riheng Jia\corref{cor1}}
\ead{rihengjia@zjnu.edu.cn}

\cortext[cor1]{Corresponding author}

\affiliation[inst1]{
    organization={School of Big Data and Artificial Intelligence, Fuyang University of Technology},
    city={Fuyang},
    postcode={236000},
    state={Anhui},
    country={China}
}

\affiliation[inst2]{
    organization={School of Computer Science and Technology, Zhejiang Normal University},
    city={Jinhua},
    postcode={321004},
    state={Zhejiang},
    country={China}
}

\begin{abstract}
The residual maximum matching problem asks for a maximum matching whose
deletion leaves a graph with the largest possible matching number. For
unweighted graphs, previous work provides a $1/2$-approximation guarantee
in general graphs and a $2/3$ guarantee in graphs with a perfect matching.
We present a $2/3$-approximation algorithm running in $O(n^5)$ time,
extending the latter guarantee to general graphs. The algorithm improves
the residual matching number through single-edge exchanges involving
free vertices. We prove that, for any maximum matching that admits no
further improvement by such exchanges, the residual matching number is
at least $2/3$ of the optimum, and we give an example showing that this
bound is tight.
\end{abstract}

\end{frontmatter}

\setcounter{section}{-1}
\section{Introduction}
\label{sec:introduction}

An edge coloring partitions the edges of a graph into matchings, with
each color class consisting of edges that have no common endpoints.
This structure also describes pairwise tasks performed in successive
stages: tasks in the same stage cannot share an endpoint, whereas
vertex resources can be reused in different stages. When two stages
have strict priorities, the first stage must accommodate as many tasks
as possible, and the number of tasks in the second stage is then
maximized over the remaining edges. The resulting optimization problem
asks for a maximum matching whose deletion leaves a graph with the
largest possible matching number.

In their work on edge coloring, Gabow, Nishizeki, Kariv, Leven, and
Terada introduced the lexicographical matching problem, which
lexicographically maximizes the cardinalities of two edge-disjoint
matchings, and established related hardness results~\cite{gabow1985}.
This model is equivalent to the residual maximum matching problem.
Kamalian and Mkrtchyan further studied residual matching numbers after
the deletion of a maximum matching and proved that the associated
optimization problems remain NP-hard even for connected bipartite
graphs of maximum degree three~\cite{kamalian2008}. Once the first
matching is fixed, a maximum matching of the residual graph can be
found in polynomial time; the difficulty lies in choosing the first
maximum matching.

Let $\nu(G)$ denote the matching number of $G$, and let $L(G)$ and
$l(G)$ denote, respectively, the largest and smallest residual matching
numbers obtained by deleting a maximum matching. Khojabaghyan and
Mkrtchyan proved that $L(G)\leq 2l(G)$ for every graph and that
$L(G)\leq 3l(G)/2$ when $G$ has a perfect matching~\cite{khojabaghyan2012}.
Thus, an arbitrary maximum matching gives a $1/2$-approximation in
general graphs and a $2/3$-approximation in graphs with a perfect
matching. Mkrtchyan subsequently studied the approximation complexity
of these parameters~\cite{mkrtchyan2025}, using the above $1/2$ guarantee
for general graphs and proving that, unless $\mathrm{P}=\mathrm{NP}$,
for every $0<\varepsilon<1/88$ there is no polynomial-time
$(1-\varepsilon)$-approximation algorithm for $L(G)$ on connected
bipartite graphs.

The maximum $2$-edge-colorable subgraph problem, in contrast, maximizes
the total number of edges in two edge-disjoint matchings.
Kosowski gave a $5/6$-approximation algorithm~\cite{kosowski2009}, and
Chen, Konno, and Matsushita improved the guarantee to approximately
$0.842$~\cite{chen2010}. Chen, Fan, and Wang further studied
parameterized algorithms for finding two edge-disjoint matchings and
approximation algorithms for their total weight~\cite{chen2014}.
These results optimize the combined size or weight of the two
matchings, whereas the lexicographical model requires the first
matching to have cardinality exactly $\nu(G)$. A maximum
$2$-edge-colorable subgraph need not contain a maximum matching of the
original graph~\cite{mkrtchyan2025}. Consequently, approximation
guarantees for the combined objective do not directly yield guarantees
for $L(G)$.

We improve the above $1/2$ baseline for general graphs to $2/3$. Our
algorithm uses single-edge exchanges involving free vertices: an edge
of the current maximum matching is replaced by an adjacent edge whose
other endpoint is free, and the exchange is accepted only if it
increases the residual matching number. Each accepted exchange
preserves the maximum cardinality of the first matching and increases
the residual matching number by one. Using standard maximum matching
algorithms~\cite{edmonds1965,galil1986}, the procedure runs in $O(n^5)$
time and $O(n^2)$ space. The key observation is that, at termination,
every edge of the maximum matching whose individual addition to the
residual graph increases its matching number, called a \emph{critical
edge}, has neither endpoint adjacent to a free vertex.
Combining this structural property with
a deficit bound for noncritical edges and an assignment argument on
oriented alternating paths and cycles, we establish the $2/3$
approximation guarantee, extending the ratio for graphs with a perfect
matching to general graphs. A seven-vertex example shows that the
$2/3$ bound implied by this termination condition is tight.

\section{Problem definition and preliminaries}
\label{sec:preliminaries}

\subsection{Problem definition}

Let $G=(V,E)$ be a finite simple undirected graph with $n=|V|$ and
unweighted edges. Write $\nu(G)$ for the cardinality of a maximum
matching of $G$, and let $\mathcal{M}(G)$ be the set of all maximum
matchings of $G$. For an edge set $S\subseteq E$, let $V(S)$ denote
the set of its endpoints, and let $G-S$ denote the graph obtained
by deleting only the edges of $S$ and retaining all vertices. For
an edge $e$ not already present in a graph $G$, let $G+e$ denote
the graph obtained by adding $e$.

The objective is to maximize the residual matching number after
deleting a maximum matching:
\begin{equation}
L(G)=\max_{M\in\mathcal{M}(G)}\nu(G-M).
\label{eq:objective}
\end{equation}
Different maximum matchings may yield different residual matching
numbers. For example, the five-vertex path in
Fig.~\ref{fig:residual} has $\nu(G)=2$, but deleting the maximum
matchings $M=\{12,45\}$ and $M'=\{23,45\}$ gives
$\nu(G-M)=1$ and $\nu(G-M')=2$, respectively.

\begin{figure}[htbp]
\centering
\includegraphics[width=\linewidth]{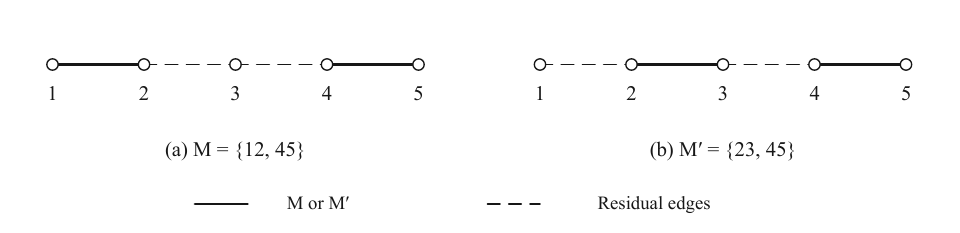}
\caption{Residual graphs corresponding to different maximum matchings.}
\label{fig:residual}
\end{figure}

\subsection{Symmetric difference and alternating paths}

For two matchings $A$ and $B$, their symmetric difference is
$A\symdiff B=(A\setminus B)\cup(B\setminus A)$. In the graph
formed by these edges, every connected component containing an
edge is an alternating path or an even cycle, with edges of
$A\setminus B$ and $B\setminus A$ appearing alternately.
If the first and last edges of a path both belong to
$A\setminus B$, then the path contains one more $A$-edge than
$B$-edge. Replacing the $B$-edges on the path by its $A$-edges
therefore yields a matching with one more edge than
$B$~\cite{edmonds1965}.

In particular, if $A$ and $B$ are both maximum matchings of the same
graph, every path component must contain equal numbers of the two
types of edges. Otherwise, exchanging the two types along that path
would increase one of the matchings, contradicting its maximum
cardinality. Thus, every connected component of $A\symdiff B$
containing an edge is an alternating path of even length or an
alternating cycle, with equal numbers of $A$-edges and
$B$-edges~\cite{galil1986}.

\section{The single-edge exchange algorithm}
\label{sec:algorithm}

\subsection{The exchange criterion}

Let $M$ be the current maximum matching, and let $H=G-M$ and
$\mu=\nu(H)$. The vertices in $V\setminus V(M)$ are called
\emph{$M$-free vertices}. An edge $e\in M$ is \emph{critical
with respect to $H$} if $\nu(H+e)=\mu+1$, and \emph{noncritical}
if $\nu(H+e)=\mu$.

Let $u$ be an $M$-free vertex, and suppose that $vw\in M$ and
$uv\in E$. Set
\begin{equation}
M'=(M\setminus\{vw\})\cup\{uv\},
\qquad G-M'=H+vw-uv.
\label{eq:exchange}
\end{equation}
Since $u$ is free, $M'$ is a matching of the same cardinality as
$M$ and is therefore also maximum. Figure~\ref{fig:exchange}
illustrates this exchange.

\begin{figure}[htbp]
\centering
\includegraphics[width=\linewidth]{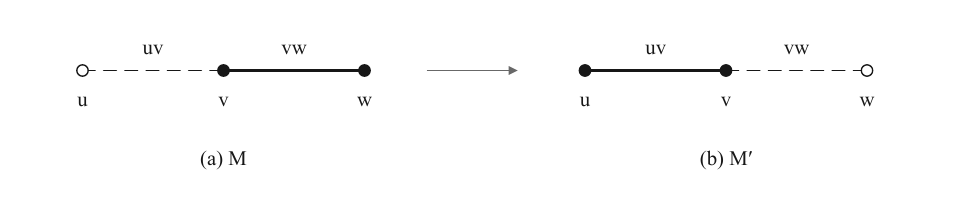}
\caption{The single-edge exchange $vw\to uv$. Vertex $u$ is free
before the exchange, and $w$ is free afterwards. Thick solid lines
represent edges of the current maximum matching; dashed lines
represent residual edges.}
\label{fig:exchange}
\end{figure}

\begin{lemma}
\label{lem:exchange}
For the exchange in~\eqref{eq:exchange},
$\nu(G-M')=\mu+1$ if and only if $vw$ is critical.
\end{lemma}

\begin{proof}
Adding one edge increases the matching number by at most one.
If $vw$ is critical, then $\nu(H+vw)=\mu+1$. Take a matching
$N^+$ of cardinality $\mu+1$ in $H+vw$. It must contain $vw$;
otherwise it would be a matching in $H$, contradicting
$\nu(H)=\mu$. Since $uv$ and $vw$ share the vertex $v$, the
matching $N^+$ does not contain $uv$ and is therefore also a
matching in $H+vw-uv=G-M'$. Together with
$G-M'\subseteq H+vw$, this gives
\[
\mu+1\leq \nu(G-M')\leq \nu(H+vw)=\mu+1.
\]
Hence $\nu(G-M')=\mu+1$.

Conversely, if $\nu(G-M')=\mu+1$, then the inclusion
$G-M'\subseteq H+vw$ gives
\[
\mu+1=\nu(G-M')\leq \nu(H+vw)\leq\mu+1.
\]
Thus $\nu(H+vw)=\mu+1$, so $vw$ is critical.
\end{proof}

\subsection{The algorithm and its termination property}

Algorithm~\ref{alg:exchange} starts with an arbitrary maximum
matching and repeatedly performs the improving exchanges described
in Lemma~\ref{lem:exchange}.

\begin{algorithm}[htbp]
\caption{Single-edge exchange algorithm}
\label{alg:exchange}
\begin{minipage}{\linewidth}
\textbf{Input:} A finite simple undirected graph $G$.\\
\textbf{Output:} A maximum matching $M$ of $G$ and a maximum
matching $N$ of the residual graph $G-M$.
\begin{enumerate}
\item Find an arbitrary maximum matching $M$ of $G$.
\item Set $H=G-M$. Find a maximum matching $N$ of $H$, and set
      $\mu=|N|$ and $U=V\setminus V(M)$.
\item Scan the edges $vw\in M$ in a fixed order. If neither endpoint
      has a neighbor in $U$, continue the scan.
\item Otherwise, choose any vertex $u$ with $u\in U$ and $uv\in E$,
      and compute $\nu(H+vw)$.
\item If $\nu(H+vw)=\mu+1$, set
      $M\leftarrow(M\setminus\{vw\})\cup\{uv\}$ and return to Step~2.
\item If the scan ends without a successful exchange, output $(M,N)$.
\end{enumerate}
\end{minipage}
\end{algorithm}

\begin{proposition}
\label{prop:termination}
Algorithm~\ref{alg:exchange} performs at most $\lfloor n/2\rfloor$
successful exchanges. Its output $M$ is a maximum matching, and
every edge of $M$ with an endpoint adjacent to a vertex in $U$
is noncritical; that is,
\begin{equation}
u\notin V(M),\quad uv\in E,\quad vw\in M
\quad\Longrightarrow\quad \nu(H+vw)=\mu.
\tag{S}\label{eq:S}
\end{equation}
\end{proposition}

\begin{proof}
By Lemma~\ref{lem:exchange}, each successful exchange preserves
the maximum cardinality of $M$ and increases the integer $\mu$
by one. Since $\mu\leq\lfloor n/2\rfloor$, the number of exchanges
is bounded. At termination, every matching edge with an endpoint
adjacent to a free vertex has been tested without yielding an improvement.
Thus property~\eqref{eq:S} holds.
\end{proof}

Equivalently, property~\eqref{eq:S} states that neither endpoint
of any critical edge in $M$ has a neighbor in $U$.

\section{A deficit bound for noncritical edges}
\label{sec:deficit}

\begin{lemma}
\label{lem:deficit}
Let $M$ be a maximum matching of a graph $G$, and let $H=G-M$
and $\mu=\nu(H)$. Let $N_0$ be a matching of cardinality $\mu-t$
in $H$, and let $Z\subseteq M$ consist of edges that are
noncritical with respect to $H$. If
$V(Z)\cap V(N_0)=\varnothing$, then
\begin{equation}
|Z|\leq 2t.
\label{eq:deficit}
\end{equation}
\end{lemma}

\begin{proof}
Take a maximum matching $N$ of $H$ and set $R=N_0\cup Z$.
By the vertex-disjointness assumption, $R$ is a matching and
$|R|=\mu-t+|Z|$. In $R\symdiff N$, let $p$ count the path
components with one more $R$-edge than $N$-edge, and let $q$
count those with one more $N$-edge than $R$-edge. The remaining
components do not affect the difference in cardinalities, so
\begin{equation}
|Z|-t=|R|-|N|=p-q\leq p.
\label{eq:pathdifference}
\end{equation}

Each path counted by $p$ contains at least two edges of $Z$.
Indeed, if it contains no edge of $Z$, as in
Fig.~\ref{fig:deficit}(a), then the entire path lies in $H$
and is an $N$-augmenting path, contradicting the maximum
cardinality of $N$. If it contains exactly one edge $e\in Z$,
as in Fig.~\ref{fig:deficit}(b), then the entire path lies in
$H+e$, and augmenting along it yields a matching of cardinality
$\mu+1$. However, $e$ is noncritical, so $\nu(H+e)=\mu$,
a contradiction.

These $p$ paths are pairwise vertex-disjoint, and each contains
at least two edges of $Z$. They therefore contain at least $2p$
distinct edges of $Z$ in total, giving $2p\leq |Z|$.
Together with~\eqref{eq:pathdifference}, this yields
$|Z|-t\leq |Z|/2$, which proves the claim.
\end{proof}

\begin{figure}[htbp]
\centering
\includegraphics[width=\linewidth]{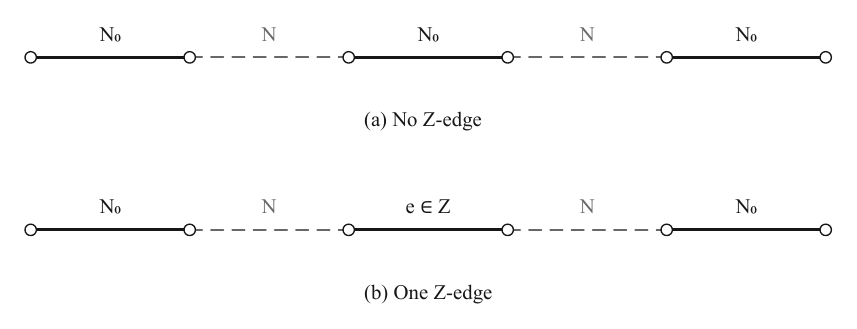}
\caption{Augmenting paths in Lemma~\ref{lem:deficit}.}
\label{fig:deficit}
\end{figure}

\section{The approximation guarantee for general graphs}
\label{sec:guarantee}

Throughout this section, let $M$ be a maximum matching satisfying
property~\eqref{eq:S}, and write $H=G-M$ and $\mu=\nu(H)$.
Take an optimal matching pair $(F,Q)$, where $F$ is a maximum
matching of $G$, $Q$ is a maximum matching of $G-F$, and
$|Q|=L(G)$. For brevity, write $L=L(G)$.

\subsection{The comparison matching and noncritical conflict edges}

Define
\begin{equation}
K=M\cap Q,\qquad c=|K|,\qquad
a=|M\setminus F|=|F\setminus M|.
\label{eq:K}
\end{equation}
We call $K$ the \emph{conflict edge set}. Since
$F\cap Q=\varnothing$, we have $K\subseteq M\setminus F$.
Moreover, $F\setminus M$ is a matching in $H$, so
\begin{equation}
c\leq a\leq\mu.
\label{eq:a}
\end{equation}
Let $r=|(M\setminus F)\setminus K|=a-c$.
Figure~\ref{fig:sets} illustrates the relationships among
these edge sets.

\begin{figure}[htbp]
\centering
\includegraphics[width=\linewidth]{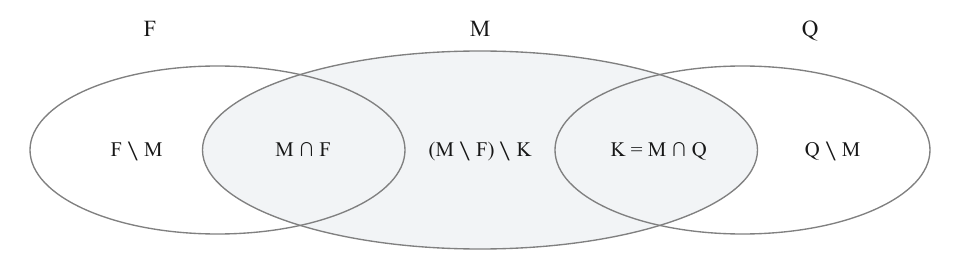}
\caption{Relationships among the matching edge sets.}
\label{fig:sets}
\end{figure}

Deleting the conflict edges from $Q$ gives a matching of
cardinality $L-c$ in $H$. Following the comparison matching
construction used in~\cite{khojabaghyan2012} to prove the
bound for graphs with a perfect matching, let $W=V(K)$ and set
\begin{equation}
X=\{f\in F:V(\{f\})\subseteq W\},\qquad x=|X|.
\label{eq:X}
\end{equation}
We have $X\cap M=\varnothing$. Indeed, suppose that
$f\in X\cap M$. Since the endpoints of $f$ are already
covered by $K\subseteq M$, the matching property of $M$
implies $f\in K\subseteq Q$, contradicting
$F\cap Q=\varnothing$. Furthermore, $Q\setminus M=Q\setminus K$.
Because $Q$ is a matching, the edges of $Q\setminus K$ have
no common endpoints with the edges of $K$, and hence all
their endpoints lie outside $W=V(K)$. Both endpoints of
every edge of $X$, on the other hand, lie in $W$.
Thus $Q\setminus M$ and $X$ are vertex-disjoint. As both
are matchings in $H$,
\begin{equation}
N_0=(Q\setminus M)\cup X
\label{eq:N0}
\end{equation}
is a matching in $H$, satisfying
\begin{equation}
|N_0|=L-c+x,\qquad t=\mu-(L-c+x)\geq 0.
\label{eq:t}
\end{equation}
Figure~\ref{fig:construction} illustrates this construction:
after deleting the conflict edge set $K$ from $Q$, we add
the edges of $X$ between endpoints of $K$ to obtain the
matching $N_0$ in $H$.

\begin{figure}[htbp]
\centering
\includegraphics[width=\linewidth]{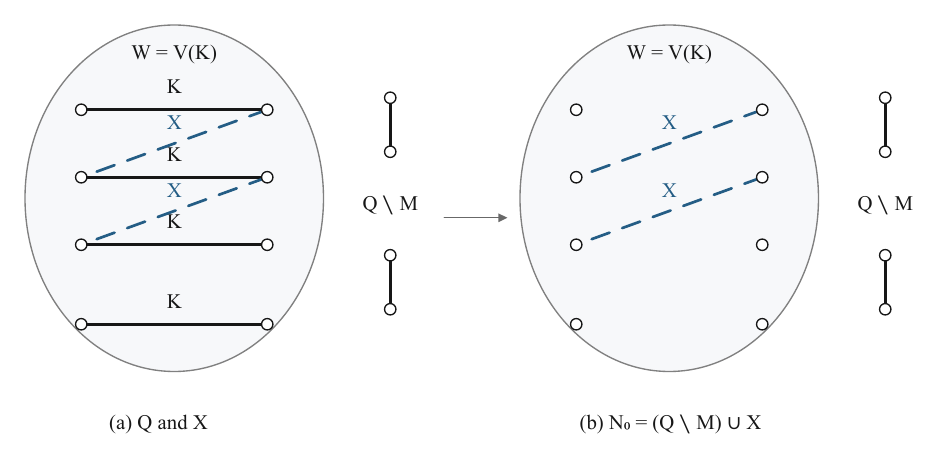}
\caption{Construction of the matching $N_0$. Black solid lines
represent edges of $Q$, and blue dashed lines represent
$X\subseteq F\setminus M$.}
\label{fig:construction}
\end{figure}

Let $Z$ be the set of noncritical edges in $K$ that are
vertex-disjoint from $X$, and write $z=|Z|$.

\begin{lemma}
\label{lem:z}
The sets defined above satisfy
\begin{equation}
z\leq 2t.
\label{eq:z}
\end{equation}
\end{lemma}

\begin{proof}
Since $Z\subseteq K\subseteq M$, the set $Z$ is a matching
disjoint from $E(H)$, and each of its edges is noncritical
with respect to $H$. As $K\subseteq Q$ and $Q$ is a matching,
$Z$ is vertex-disjoint from $Q\setminus M$; by definition,
$Z$ is also vertex-disjoint from $X$. Hence
$V(Z)\cap V(N_0)=\varnothing$. Applying
Lemma~\ref{lem:deficit} to $N_0$ in~\eqref{eq:N0}
proves the claim.
\end{proof}

\subsection{An assignment argument on oriented alternating components}

The remaining conflict edges $K\setminus Z$ consist of
critical edges and noncritical edges sharing an endpoint
with an edge of $X$.
We use the alternating structure of $M\symdiff F$ to
bound their number.

\begin{lemma}
\label{lem:assignment}
The sets defined above satisfy
\begin{equation}
c-z\leq 2x+r.
\label{eq:assignment}
\end{equation}
\end{lemma}

\begin{proof}
Since $M$ and $F$ are both maximum matchings, every
component of $M\symdiff F$ containing an edge is an
alternating path of even length or an even cycle.
Orient each path from its $F$-free endpoint to its
$M$-free endpoint, and choose either cyclic orientation
for each cycle. Then each edge of $M\setminus F$ is
immediately followed by a unique adjacent edge of
$F\setminus M$, as illustrated in Fig.~\ref{fig:oriented}.

\begin{figure}[htbp]
\centering
\includegraphics[width=\linewidth]{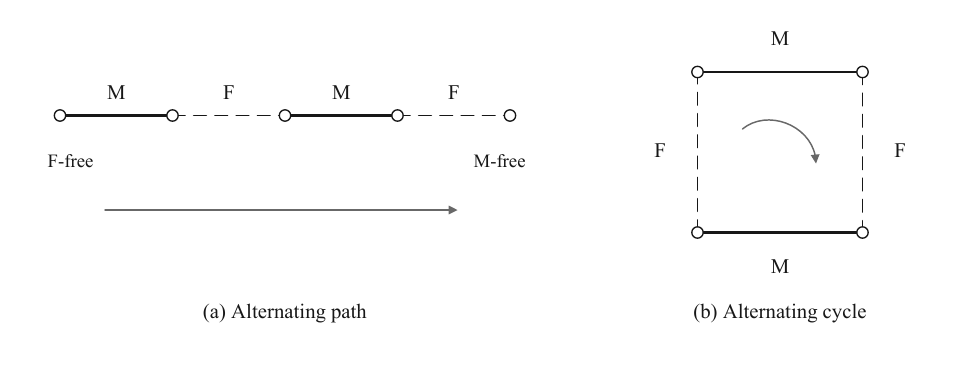}
\caption{Oriented alternating components. Thick solid lines
represent edges of $M\setminus F$, and dashed lines represent
edges of $F\setminus M$. Each path is oriented from its
$F$-free endpoint to its $M$-free endpoint, and each cycle
is given a fixed cyclic orientation.}
\label{fig:oriented}
\end{figure}

For a critical edge $e\in K$, let $f$ be the edge of
$F\setminus M$ immediately following it. There must be
an edge $e'\in M\setminus F$ immediately after $f$.
On a cycle, this follows from the alternating structure.
On a path, if $f$ were the last edge, it would join an
endpoint of $e$ to an $M$-free vertex; property~\eqref{eq:S}
would then force $e$ to be noncritical, a contradiction.
Thus $e,f,e'$ are three consecutive edges. If $e'\in K$,
both endpoints of $f$ lie in $W$, so $f\in X$; assign
$e$ to $f$. If $e'\notin K$, assign $e$ to
$e'\in(M\setminus F)\setminus K$.

For a noncritical edge $e\in K\setminus Z$, the definition
ensures that $e$ shares an endpoint with at least one edge of $X$;
assign $e$ to any such edge. Every edge of $K\setminus Z$
has now been assigned exactly once. The three cases are
illustrated in Fig.~\ref{fig:assignment}.

\begin{figure}[htbp]
\centering
\includegraphics[width=.94\linewidth]{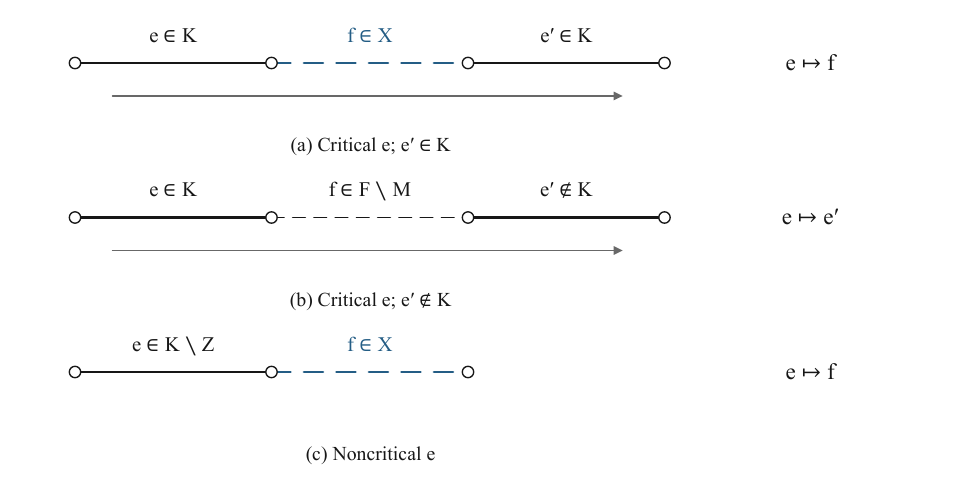}
\caption{Assignment of conflict edges.}
\label{fig:assignment}
\end{figure}

Each edge $f\in X$ receives assignments only from edges
of $K$ sharing an endpoint with it. Since $K$ is a matching, each
endpoint of $f$ is incident to at most one edge of $K$,
so $f$ receives at most two assignments. An edge
$e'\in(M\setminus F)\setminus K$ receives an assignment
only under the second rule for critical edges, which
requires the consecutive order $e,f,e'$. Once the
orientation is fixed, the predecessor $f$ of $e'$ and
the predecessor $e$ of $f$ on the path or cycle are
both unique. Hence $e'$ receives at most one assignment.
The two types of receiving edges accommodate at most
$2x+r$ assignments in total, and therefore
$|K\setminus Z|=c-z\leq 2x+r$.
\end{proof}

\subsection{The main theorem}

\begin{theorem}
\label{thm:guarantee}
Let $M$ be a maximum matching of $G$. If $M$ satisfies
property~\eqref{eq:S}, then
\begin{equation}
\nu(G-M)\geq \frac{2}{3}L(G).
\label{eq:guarantee}
\end{equation}
\end{theorem}

\begin{proof}
By Lemmas~\ref{lem:z} and~\ref{lem:assignment},
\begin{equation}
c\leq z+2x+r\leq 2t+2x+r.
\label{eq:combine}
\end{equation}
Substituting $t=\mu-(L-c+x)$ and $r=a-c$ gives
\[
c\leq 2[\mu-(L-c+x)]+2x+(a-c)
 =a+c+2\mu-2L.
\]
Using $a\leq\mu$ from~\eqref{eq:a}, we obtain
\begin{equation}
2L\leq a+2\mu\leq 3\mu.
\label{eq:finalbound}
\end{equation}
Thus $\mu\geq 2L/3$.
\end{proof}

\section{Running time and tightness}
\label{sec:complexity}

\subsection{Running time}

\begin{theorem}
\label{thm:runtime}
With a standard implementation of maximum matching in
general graphs, Algorithm~\ref{alg:exchange} runs in
$O(n^5)$ time and $O(n^2)$ space.
\end{theorem}

\begin{proof}
A maximum matching in a general graph can be found
in $O(n^3)$ time and $O(n^2)$ space~\cite{edmonds1965,galil1986}.
Each scan considers at most $\lfloor n/2\rfloor$
matching edges. For each edge, finding a free neighbor
takes $O(n)$ time, and the maximum matching subroutine
is called at most once to compute the matching number
after adding that edge back. Thus each scan takes
$O(n^4)$ time. By Proposition~\ref{prop:termination},
there are at most $O(n)$ successful exchanges.
Initialization and the recomputation of a residual
maximum matching in each round are also covered by
this bound. Hence the total running time is
\begin{equation}
O(n)\cdot O(n^4)=O(n^5).
\label{eq:time}
\end{equation}
An adjacency-matrix representation of the graph,
together with sequential reuse of the workspace
of the maximum matching subroutine, requires
$O(n^2)$ total space.
\end{proof}

\subsection{Tightness of the termination condition}

\begin{proposition}
\label{prop:tight}
There exist a graph $G$ and a maximum matching $M$
satisfying property~\eqref{eq:S} such that
$\nu(G-M)=2L(G)/3$. Thus the $2/3$ approximation
bound implied by property~\eqref{eq:S} alone is tight.
\end{proposition}

\begin{proof}
Let $V=\{0,1,2,3,4,5,6\}$ and
\begin{equation}
E=\{01,04,12,14,23,45,56\}.
\label{eq:example}
\end{equation}
Set $M=\{14,23,56\}$, as shown in Fig.~\ref{fig:tight}(a).
Since $|V|=7$, we have $\nu(G)\leq 3$, so $M$ is
maximum. The nonisolated part of the residual graph
$H=G-M$ is the path $2-1-0-4-5$, and thus $\nu(H)=2$.

\begin{figure}[htbp]
\centering
\includegraphics[width=\linewidth]{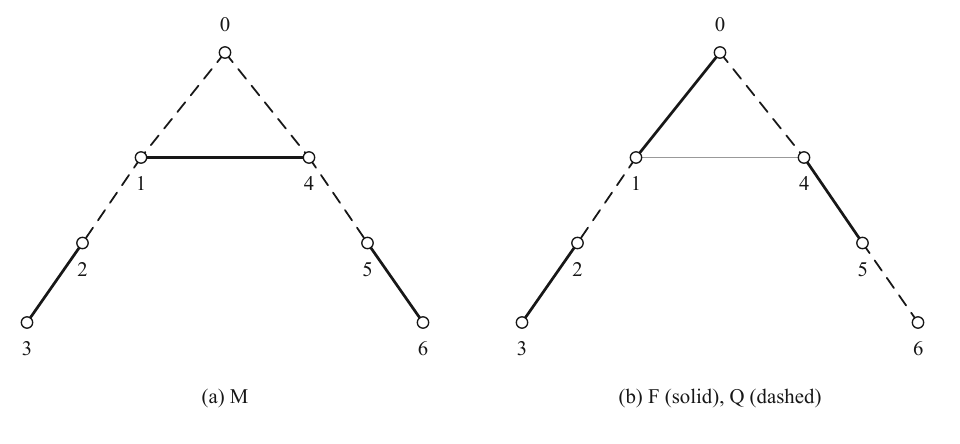}
\caption{A terminal instance attaining the $2/3$ bound.}
\label{fig:tight}
\end{figure}

The only $M$-free vertex is $0$, whose only neighbors
are $1$ and $4$. Hence $14$ is the only matching edge
that needs to be tested. The nonisolated vertices
of $H+14$ are still $0,1,2,4,5$, so its matching
number is at most two. Consequently,
$\nu(H+14)=\nu(H)=2$, and property~\eqref{eq:S}
holds. On the other hand, $F=\{01,23,45\}$ and
$Q=\{04,12,56\}$ are two edge-disjoint matchings,
each with three edges, as shown in Fig.~\ref{fig:tight}(b).
Thus $L(G)=3$, and $\nu(G-M)=2L(G)/3$.
\end{proof}

\section{Conclusion}

We have presented a $2/3$-approximation algorithm for
the residual maximum matching problem in general
unweighted graphs. Through single-edge exchanges
involving free vertices, the algorithm obtains a
maximum matching satisfying a local termination
condition. A deficit bound for noncritical edges
and an assignment argument on oriented alternating
components together establish the approximation
guarantee implied by this condition, and an example
attains the bound. Further research may investigate
whether stronger local improvement conditions can
yield a ratio better than $2/3$, and how the
structural analysis developed here can be extended
to the residual maximum matching problem in
weighted graphs.

\section*{Declaration of generative AI and AI-assisted technologies in the manuscript preparation process}

The authors used an AI assistant for translation and language revision, manuscript organization, mathematical discussion, and reference checking. The authors reviewed and edited the content as needed and assume full responsibility for all content.

\end{document}